\documentclass[12pt,a4paper]{amsart}
\usepackage[cmtip,arrow]{xy}
\usepackage{pb-diagram,pb-xy}
\usepackage{amsxtra}
\usepackage{enumitem, fancyref, theoremref}
\usepackage{amsmath, amssymb, tikz, amsthm, thmtools}
\usepackage{amsrefs, latexsym,pstricks}
\usepackage{hyperref, caption}

\calclayout

\DeclareMathSymbol{\subsetneq}{\mathrel}{AMSb}{"28}
\DeclareMathSymbol{\rightrightarrows}{\mathrel}{AMSa}{"13}

\newcommand{\BR}{\mathbb{R}} 
\newcommand{\BN}{\mathbb{N}} 

\newcommand{\V}{\mathbf{v}} 

\theoremstyle{plain}

\newtheorem{theorem}{Theorem}[section]
\newtheorem{lemma}{Lemma}[section]

\theoremstyle{remark}

\theoremstyle{definition}
\newtheorem{definition}{Definition}

\begin{document}

\title{Proving Parikh's theorem using Chomsky-Sch\"utzenberger theorem}
\author{Dmitry Golubenko}

\address{Faculty of Mathematics, Higher School of Economics, 6~Usacheva str., Moscow, Russia, 117312}
\email{golubenko@mccme.ru}

\begin{abstract}
	Parikh theorem was originally stated and proved in \cite{parikh}. Many different proofs of this classical theorems were produced then; our goal is to give another proof using Chomsky-Sch\"utzenberger representation theorem. We present the proof which doesn't use any formal language theory tool at all except the representation theorem, just some linear algebra. 
\end{abstract}

\maketitle

\tableofcontents

\section{Introduction}
Recall that the {\it context-free grammar} is a 4-tuple $G = (N, \Sigma, P, S)$ where $N$ is a set of nonterminals, $\Sigma$ is an alphabet, $P$ is a collection of productions of kind $X \rightarrow \gamma$ for $X \in N,~~~ \gamma \in (N \cup \Sigma)^*$, and $S$ is an initial nonterminal. We say that $w \in \Sigma^*$ is {\it derived} in CFG $G$ if there exists a sequence $\{s_i\}_{i \in [1;N]},~~~s_i \in (N \cup \Sigma)^*$ such that $s_0 = S,~~~ s_N = w$ and every $s_i$ is obtained from $s_{i-1}$ via applying some production from $P(G)$. The set of all words $w \in \Sigma^*$ derived in $G$ is called the {\it language generated by} $G$ and is denoted by $L(G)$. The formal language $L \subset \Sigma^*$ is called {\it context-free} if $L = L(G)$ for some context-free grammar $G$.

Let $\Sigma = \{a_1, \ldots a_k\}$. We introduce the {\it Parikh map} $\psi: \Sigma^* \rightarrow \BN^k$:
\begin{equation}
\psi(w) = (\#_{a_1}(w), \ldots, \#_{a_k}(w)),
\end{equation}
where $\#_{x}(w)$ is number of occurences of letter $x$ in word $w$. Clearly, this is a monoid homomoprhism since \label{addit} $\psi(w_1 w_2) = \psi(w_1) + \psi(w_2),~~~ \psi(\epsilon) = \mathbf{0}$ where $\mathbf{0} = (0, \ldots, 0)$. For any language $L$ we define {\it Parikh image} of $L$ $\psi(L)$ by
\begin{equation*}
\psi(L) = \{\psi(w)|w \in L\} \subset \BN^k
\end{equation*}
\begin{definition}
	A subset $S \subset \BN^k$ is called {\it linear} if it's a coset of a finitely generated submonoid of $\BN^k$, i. e. there exist $\V_0, \V_1, \ldots \V_N$ such that 
	\begin{equation}
	S = \{\V_0 + \sum\limits_{i=1}^{N} \lambda_i \V_i|\forall i \in [1;N]~~~\lambda_i \in \BN\}
	\end{equation}
	We denote this by $S = \V_0 + \Big\langle \V_1, \ldots \V_N \Big\rangle$.
\end{definition}
\begin{definition}
	A subset $S \in \BN^k$ is {\it semilinear} if it's a union of finite number of linear subsets.
\end{definition}
The following theorem is the classical one in the formal language theorem:
\begin{theorem}[Parikh]
	\label{parikh}
	For any context-free language $L$, $\psi(L)$ is semilinear.
\end{theorem}
The theorem was originally stated and proved by Rokhit Parikh in \cite{parikh}. Parikh's strategy was to present all parse trees as a union of a finite number of classes each having one minimal tree and all other obtained from that minimal by insertion of a {\it pump} tree.

We're gonna prove the theorem \ref{parikh} using another classical theorem:
\begin{theorem}[Chomsky-Sch\"utzenberger]
	Every context-free language $L \subset \Sigma^*$ admits the following presentation
	\begin{equation}
	L = h(Dyck_N \cap R),
	\end{equation}
	where $Dyck_N$ is a language of balanced parentheses of $N$ types, $R$ is a regular language and $h: \{(_1, )_1, \ldots (_N, )_N\}^* \rightarrow \Sigma^*$ is a language homomorphism.
\end{theorem}
Shamir tried to give a proof using that theorem, however, afterwards he regarded in \cite{goldstine} that his proof was fallacious. We, however, prefer much easier way to obtain our goal. We investigate Chomsky-Sch\"utzenberger construction of $h$, $R$ and $N$, prove directly that $\psi(Dyck_N \cap R)$ semilinear and thus derive the main theorem.

\section{Semilinearity of $Dyck_N \cap R$}
We use the proof of Chomsky-Sch\"utzenberger theorem given in \cite{kozen}. Consider the context-free grammar $G = (N, \Sigma, P, S)$ in Chomsky normal form such that $L(G) = L \setminus \{\epsilon\}$. Define $\widetilde{G} = (N, \widetilde{\Sigma}, \widetilde{P}, S)$, where
\begin{align*}
\widetilde{\Sigma} = \{(_p, )_p, [_p, ]_p | p \in P \}
\end{align*}
and
\begin{align*}
\widetilde{P} = \{ A \rightarrow (_p B )_p [_p C ]_p \mid p = A \rightarrow BC \} \bigcup \{ A \rightarrow (_p  )_p [_p  ]_p \mid p = A \rightarrow a \}.
\end{align*}
Thus we obtain $L(\widetilde{G}) = Dyck_{2|P|} \cap R_L$, where $R \subset \widetilde{\Sigma}^*$ is a regular language defined by the following conditions:
\begin{enumerate}
	\item every $)_p$ is immediately followed by $[_p$;
	\item no $]_p$ is immediately followed by any $(_i$ or $[_i$;
	\item if $p = A \rightarrow BC$, then $(_p$ is immediately followed by $(_q$ for some $q \in P$ and $[_p $ is immediately followed by $(_r$ for some $r \in P$;
	\item if $p = A \rightarrow a$, then $(_p )_p [_p ]_p$ is a subword, i. e. $)_p$ immediately follows $(_p$, $[_p$ follows $)_p$ and so on;
	\item the first letter is $(_p$ for some $p$ with $S$ in the left side.
\end{enumerate}
It is proven in \cite{kozen} that $L(\widetilde{G})$ is actually equal to $Dyck_{2|P|} \cap R_L$. We'll prove now that $\psi(Dyck_{2|P|} \cap R_L)$ is semilinear.

\begin{lemma}
	For $G$ and $R_L$ defined above we have that $\psi(Dyck_{2|P|} \cap R_L)$ is semilinear.
\end{lemma}
\begin{proof}
	First, we have that $\#_{(_p}(w) = \#_{)_p}(w)$ and $\#_{[_p}(w) = \#_{]_p}(w)$ for every $w \in Dyck_{2|P|} \cap R_L$ and $p \in P$. Denote $\lambda_p(w) = \#_{(_p}(w)$. From condition (1) we have $\#_{p[_p}(w) = \#_{)_p}(w) = \#_{(_p}(w)$ for every production $p$, so $\lambda_p(w)$ is the number of occurences of $p$ in derivation of $w$. We prove now that $$\Omega = \{\Big(\lambda_{p_1}(w), \ldots \lambda_{p_{|P|}}\Big) \mid  w \in Dyck_{2|P|} \cap R_L, \, p \in P\} \subset \BR^{|P|}$$ is semilinear and thus establish semilinearity of $\psi(Dyck_{2|P|} \cap R_L)$.
	
	For parse tree $t$ of $\widetilde{G}$ and $p \in P$ define $\mu_p$ as the number of occurences of $p$ in $t$. Define parse tree $t$ to be {\it irreducible} if its derived string is $w_1 X w_2$ for $w_1, w_2 \in \widetilde{\Sigma}^*$ and $X \in N$ is nonterminal in the root of $t$ and there is no path from leaf to root containing at least two instances of some nonterminal $Y \in N$. If $\sum_{p \in P} \mu_{p}(t) \geqslant 2^{|N| + 2}$ for some parse tree $t$, then $t$ is not irreducible since its depth is at least $|N| + 2$, thus it has at least $|N| + 1$ nonterminals along it. Thus the set $I$ of irreducible parse trees is finite. For parse tree $t$ define
	\begin{align*}
	\psi(t) = \{\Big(\mu_{p_1}(w), \ldots \mu_{p_{|P|}}\Big) \mid  w \in Dyck_{2|P|} \cap R_L, \, p \in P\} \subset \BR^{|P|}.
	\end{align*}
	Define parse tree to be {\it minimal} if it doesn't contain an irreducible subtree. Since the set of irreducible trees is finite, the set of minimal trees is finite since any path from leaf to root should pass through no more than $|I|$ distinct irreducible trees and thus the depth of minimal tree is limited by $(|I| + 1)(|N| + 1)$. Thus the set of all parse trees in $\widetilde{G}$ is the union of $T_m$, the sets of parse trees obtained from $m$ by adding the irreducible trees, and since
	\begin{align*}
	\psi(T_m) = \psi(m) + \Big\langle \{ \psi(t) | t \text{ is irreducible } \} \Big\rangle
	\end{align*}
	is linear, the $\Omega$ is semilinear, an so is $\psi(Dyck_{2|P|} \cap R_L)$.
\end{proof}

\section{Proof of the main theorem}
The proof of the main theorem easily derives from the lemma above and the following proposition.
\begin{lemma}
	For every homomorphism $\phi: \Sigma^* \rightarrow \Gamma^*$ semilinearity $\psi(L)$ implies semilinearity of $\psi(\phi(L))$.
\end{lemma}
\begin{proof}
	Homorphism $\phi: \Sigma^* \rightarrow \Gamma^*$ induces the linear map
	\begin{equation*}
	\Phi: \BN^{|\Sigma|} \rightarrow \BN^{|\Gamma|},~~~ (\underbrace{0, \ldots 0}_{i-1}, 1, \underbrace{0, \ldots 0}_{|\Sigma| - i}) \mapsto (\#_1(\Phi(a_i)), \ldots \#_{|\Gamma|}(\Phi(a_i)))
	\end{equation*}
	so we just linearly change generators of linear subsets and preserve the semilinearity.
\end{proof}
\section{Acknowledgements}
I came up with this when I was preparing a home task for my MIPT second-year undergraduate students; I wanted to give a problem which could use Parikh or Chomsky-Sch\"utzenberger theorem. Then I realized that this very proof was never met in literature before. I gratefully thank those brilliant sharp-minded kids, especially from group 673, who I teach with such great delight and wish they will discover something quite powerful and elegant in the future.

I would like to thank Sergey P. Tarasov and Alexander Rubtsov for useful discussions.
 

\hbadness=1100

\end{document}